\documentclass[12pt]{amsart}

\usepackage{amsmath, amsthm, latexsym, amssymb, amsfonts, epsf, xcolor, hyperref}
\usepackage{caption, subcaption, enumerate, color}
\usepackage{multirow,multicol}
\usepackage{mathtools}
\usepackage{bm}
\usepackage{blkarray}
\usepackage{booktabs}
\usepackage{framed}
\usepackage{setspace}
\usepackage{tikz}
\usepackage{graphicx}
\usepackage{wrapfig}
\usepackage{float}
\usepackage[dvips]{epsfig}
\usepackage{bbm}
\usepackage{url}
\usepackage{booktabs}
\usepackage{bibentry}
\usepackage[normalem]{ulem}
\usepackage{comment}
\usepackage{lingmacros}
\usepackage{cleveref}
\usepackage{extarrows}
\usepackage{multirow}
\usepackage{mathtools}
\usepackage{enumerate}
\usepackage{array}
\usepackage{bigstrut}
\usepackage{exscale,relsize}
\usepackage{graphicx}
\usepackage[utf8]{inputenc}
\usepackage[linesnumbered, ruled]{algorithm2e}

\newcommand{\cal}[1]{\mathcal{#1}}

\newcommand{\F}{{\mathbb F}}

\numberwithin{equation}{section}
\newtheorem{theorem}{Theorem}[section]
\newtheorem{lemma}[theorem]{Lemma}
\newtheorem{proposition}[theorem]{Proposition}
\newtheorem{corollary}[theorem]{Corollary}

\theoremstyle{definition}
\newtheorem{definition}[theorem]{Definition} 
\newtheorem{remark}[theorem]{Remark}
\newtheorem{example}[theorem]{Example}

\newcommand{\rmv}[1]{}

\DeclareMathOperator{\wt}{wt}
\DeclareMathOperator{\hull}{Hull}
\DeclareMathOperator{\rk}{rk}
\DeclareMathOperator{\diag}{diag}

\begin{document}


\title[Relative hulls and quantum codes]{Relative hulls and quantum codes}

\author[S. Anderson]{Sarah E. Anderson}
\address[Sarah E. Anderson]{Department of Mathematics \\ University of St. Thomas \\ St. Paul, MN USA}
\email{ande1298@stthomas.edu}

\author[E. Camps-Moreno]{Eduardo Camps-Moreno}
\address[Eduardo Camps-Moreno]{Department of Mathematics\\ Virginia Tech\\ Blacksburg, VA USA}
\email{eduardoc@vt.edu}

\author[H. L\'opez]{Hiram H. L\'opez}
\address[Hiram H. L\'opez]{Department of Mathematics\\ Virginia Tech\\ Blacksburg, VA USA}
\email{hhlopez@vt.edu}

\author[G. Matthews]{Gretchen L. Matthews}
\address[Gretchen L. Matthews]{Department of Mathematics\\ Virginia Tech\\ Blacksburg, VA USA}
\email{gmatthews@vt.edu}

\author[D. Ruano]{Diego Ruano}
\address[Diego Ruano]{IMUVA-Mathematics Research Institute\\ Universidad de Valladolid\\ Valladolid, Spain}
\email{diego.ruano@uva.es}

\author[I. Soprunov]{Ivan Soprunov}
\address[Ivan Soprunov]{Department of Mathematics and Statistics\\ Cleveland State University\\ Cleveland, OH USA}
\email{i.soprunov@csuohio.edu}

\thanks{Hiram H. L\'opez was partially supported by the NSF grants DMS-2201094 and DMS-2401558. Gretchen L. Matthews was partially supported by NSF DMS-2201075 and the Commonwealth Cyber Initiative. Diego Ruano was partially supported by Grant RYC-2016-20208 funded by MCIN/AEI/10.13039/501100011033 and by ``ESF Investing in your future,'' by Grant TED2021-130358B-I00 funded by MCIN/AEI/10.13039/501100011033 and by the ``European Union NextGenerationEU/PRTR,'' and by QCAYLE project funded by MCIN, the European Union NextGenerationEU (PRTR C17.I1) and Junta de Castilla y Le\'on.}

\begin{abstract}
Given two $q$-ary codes $C_1$ and $C_2$, the relative hull of $C_1$ with respect to $C_2$ is the intersection $C_1\cap C_2^\perp$. We prove that when $q>2$, the relative hull dimension can be repeatedly reduced by one, down to a certain bound, by replacing either of the two codes with an equivalent one. The reduction of the relative hull dimension applies to hulls taken with respect to the $e$-Galois inner product, which has as special cases both the Euclidean and Hermitian inner products. We give conditions under which the relative hull dimension can be increased by one via equivalent codes when $q>2$. We study some consequences of the relative hull properties on entanglement-assisted quantum error-correcting codes and prove the existence of new entanglement-assisted quantum error-correcting maximum distance separable codes, meaning those whose parameters satisfy the quantum Singleton bound.
\end{abstract}

\maketitle

\section{Introduction}
Let $C$ be a linear code over a finite field $\F_q$. The hull of $C$ is defined by $\hull(C)=C \cap C^{\perp}$, where $C^{\perp}$ is the dual of $C$ taken with respect to the Euclidean inner product. Carlet, Mesnager, Tang, Qi, and Pellikaan proved in the seminal paper~\cite{lcd} that when $q>3$, the dimension of $\hull(C)$ can be reduced to zero, meaning that there exists {a code $C^\prime$ monomially equivalent to} $C$ such that $\dim \hull(C^\prime)=0.$ Thus, one can reduce the hull of a code when $q>3$ without changing their main parameters: length, dimension, and minimum distance.

Luo, Ezerman, Grassl, and Ling went further in~\cite{Grassl22} to prove that when $q>2$, the dimension of the Hermitian hull $\hull_h(C)=C \cap C^{\perp_h}$, where $C^{\perp_h}$ is the Hermitian dual of $C$, can be reduced to zero one by one in the sense that if $\dim \hull_h(C)>0$, then there exists {a code $C^\prime$ monomially equivalent to} $C$ such that $\dim \hull_h(C^\prime)=\dim \hull_h(C)-1$. A slight modification reveals the same result for the hull of $C$ (taken with respect to the Euclidean inner product) when $q>3$. Therefore, there exists a sequence of monomially equivalent codes $C_0,C_1,\dots, C_t=C$ such that  $\dim \hull(C_i)=i$, where $t=\dim \hull(C)$. How equivalent codes can change the hull is also studied in~\cite{10006387}.

It is well known that self-orthogonal codes with respect to the Hermitian inner product may be used to construct quantum error-correcting codes \cite{ashikhmin_knill, CRSS, ketkar}. Entanglement allows one to remove restrictions on the relationship between a code and its dual. Hence, any linear code (not necessarily self-orthogonal) may be used to define a quantum code~\cite{brun_science}. One may also use two codes $C_1, C_2 \subseteq \F_q^n$ satisfying $C_2^{\perp} \subseteq C_1$ via the now famous CSS construction \cite{CS, Steane}.  In the case of the construction of entanglement-assisted quantum error-correcting codes using linear codes $C_1, C_2 \subseteq \F_q^n$, the required number of pairs of maximally entangled qudits is given by the parameter $c = \dim (C_1) - \dim (C_1\cap C_2^\perp$) \cite{wilde}. Therefore, a key ingredient for computing $c$ is $C_1\cap C_2^\perp$, which we call the relative hull. More explicitly, the {\it relative hull of $C_1$ with respect to $C_2$} is
$$\hull_{C_2}(C_1) = C_1\cap C_2^\perp.$$
Note that the {hull} of $C$ is $\hull(C)=\hull_{C}(C)$.

In this paper, we study how equivalent codes change the relative hull. Specifically, we look for {codes $C'_1$ and $C'_2$ equivalent to} $C_1$ and $C_2$, respectively, such that the dimension of $\hull_{C'_2}(C'_1)$ is larger or smaller than that of $\hull_{C_2}(C_1)$. We first show that to increase or decrease the relative hull dimension, we only need to find an equivalent code for one of the codes. Then, we show that the relative hull with respect to Galois inner products \cite{fan_zhang,liu_pan} (which include the Euclidean and Hermitian inner products as particular cases) can be computed in terms of the Euclidean inner product, justifying the focus on the classical Euclidean inner product in this work. One of the main results of this paper is Theorem~\ref{22.10.29}, where we show that we can successively decrease the dimension of the relative hull by one via equivalent codes when $q>2$. We provide a similar result for $e$-Galois hulls. As a corollary, we can recover the analogous result in \cite{lcd} for the Euclidean inner product and in \cite{Grassl22} for the Hermitian inner product as special cases.

This paper also concerns increasing the relative hull dimension. Proposition~\ref{23.06.13} gives an upper bound for the dimension of $\hull_{C_2}(C_1)$, which sometimes also is an upper bound for $\dim \hull_{C_2^\prime}(C_1^\prime)$ for any {codes $C'_1$ and $C'_2$ equivalent to} $C_1$ and $C_2$. Theorem \ref{22.11.05} shows we can successively increase the dimension of $\hull_{C_2}(C_1)$ by one via equivalent codes up to the upper bound given in Proposition~\ref{23.06.13} when $q >2$. 

Another primary goal is to apply our results to quantum error-correcting codes. We use the standard notation $[[n,\kappa,\delta;c]]_q$ to mean that a quantum code $Q$ is a $q$-ary entanglement-assisted quantum error-correcting code (EAQECC) that encodes $\kappa$ logical qudits into $n$ physical qudits with the help of $n - \kappa - c$ ancillas and $c$ pairs of maximally entangled qudits. The {\it rate} $\rho$ and {\it net rate} ${\overline{\rho}}$ of $Q$ are respectively defined by
\[\rho:= \frac{\kappa}{n}, \qquad\qquad\qquad \overline{\rho}:= \frac{\kappa-c}{n}.\]
As stated, the relative hull dimension is linked to the required number of pairs of maximally entangled quantum states for an EAQECC. Our results concerning the relative hull demonstrate how monomially equivalent codes may be used to tailor the parameter $c$ within the specified bounds. Thus, we can reduce the required number of pairs of maximally entangled quantum states while maintaining the net rate. Hence, one has a simpler implementation with the same net rate. We show that if a quantum code obtained via the CSS construction using $C_1$ and $C_2$ is pure, then the minimum distance of the quantum code obtained via the CSS construction of some linear codes monomially equivalent to $C_1$ and $C_2$ does not decrease. Furthermore, we give conditions to obtain a pure quantum code using monomially equivalent codes. We obtain EAQECCs codes with excellent parameters by applying Theorem 3.3 to multivariate Goppa codes, filling in some gaps or improving the parameters of some of the best-known EAQECCs recently published by L. Sok~\cite{LSok}. We obtain new EAQMDS (EAQECCs whose parameters achieve the Singleton bound, so-called entanglement-assisted quantum maximum distance separable codes), by applying Theorem 3.3 to (possibly extended or double extended) generalized Reed-Solomon codes when $q>2$, $1 < n < q+1$, and $k \le n+2$. 

This paper is organized as follows. Preliminaries are given in Section \ref{prelim_section}. Section \ref{reducerelativehull}
provides results on reducing the relative hull while Section \ref{increaserelativehull}
discusses increasing the relative hull. Applications to the design of entanglement-assisted quantum error-correcting codes are in Section \ref{quantum_section}. The paper ends with a conclusion in Section \ref{conclusion_section}.

\section{Preliminaries} \label{prelim_section}

This section provides a foundation for the rest of the paper in terms of preliminary results and notation. Subsection~\ref{22.12.24} explores the relative hull with respect to the usual (Euclidean) inner product. Subsection~\ref{22.12.25} introduces the $e$-Galois relative hull, the relative hull with respect to the more recently introduced Galois inner products, among which we find the Hermitian inner product. Subsection~\ref{22.12.25} also proves that the $e$-Galois relative hulls are particular cases of the relative hulls with respect to the usual inner product. Subsection~\ref{22.12.26} reviews the primary constructions of quantum error-correcting codes used in this paper and links them to relative hulls. 

\subsection{Relative hulls and code equivalence}\label{22.12.24}
Let $\mathbb{F}_q$ be the finite field with $q$ elements. The multiplicative group $\mathbb{F}_q\setminus \{0\}$ is denoted by $\F_q^\ast$. For ${c}\in\mathbb{F}_q^n$, we denote by $\wt({c})$ the (Hamming) {\it weight} of $c$, which is the number of nonzero entries of $c$. For $S\subseteq\mathbb{F}_q^n$, we denote by $\wt(S)$ the {\it minimum} of the weights of the elements of $S\setminus\{0\}$. A {\it linear code} $C$ over $\F_q$ of length $n$ is a vector subspace of $\mathbb{F}_q^n$; we may say {\it code} for short because we only deal with linear codes. An $[n,k,d]_q$-code is a linear code over $\F_q$ of length $n$, dimension $k$ as an $\F_q$-subspace, and {\it minimum distance} $d(C) = \wt(C)$; we sometimes refer to such a code as an $[n,k]_q$-code if the minimum distance is irrelevant to the discussion. The {\it Euclidean dual} of $C$ is denoted and defined by
$$C^\perp = \left\{ {x} \in \mathbb{F}_q^n \mid {x} \cdot {c} = 0 \text{ for all } {c} \in C \right\},$$ 
where ${x} \cdot {c} = \sum_{i=1}^nx_ic_i$ is the {\it Euclidean inner product}.
Recall that $\hull(C)=C \cap C^{\perp}$. We say that $C$ is {\it self-orthogonal} if $\hull(C)=C$ and that $C$ is {\it linear complementary dual} (LCD) if $\hull(C)=\{0\}$. The set of $m \times n$ matrices with entries in $\F_q$ is denoted by $\F_q^{m \times n}$, and $\rk(M)$ denotes the rank of a matrix $M \in \F_q^{m \times n}$.  The {\it kernel} of $G \in \F_{q}^{k\times n}$ is $\ker (G) = \left\{ {x} \in \F_{q}^{n} \mid G{x}^T = 0 \right\}$. The $j$-th standard basis vector of $\mathbb{F}_q^{n}$ is $\bm{e}_j = (0,\dots, 0,1,0, \dots 0)$ where the only nonzero entry is in the $j$-th coordinate.

\begin{definition}\label{22.10.03}
Let $C_1$ and $C_2$ be two codes of the same length over $\F_q$.
We define the {\it relative hull of $C_1$ with respect to $C_2$} as
$$\hull_{C_2}(C_1) = C_1\cap C_2^\perp.$$
The {\it hull} of $C_1$ is $\hull(C_1) = \hull_{C_1}(C_1)$.
\end{definition}

Let $x$ be an element of $\hull_{C_1}(C_2)=C_1^\perp\cap C_2$ and $c$ an element of $\hull_{C_2}(C_1) = C_1\cap C_2^\perp$. As ${x} \cdot {c} = 0$, we conclude that $\hull_{C_1}(C_2) \subseteq \left( \hull_{C_2}(C_1) \right)^\perp$ (note that $(A \cap B)^\perp = A^\perp + B^\perp$). In particular, $\hull(C)$ is a self-orthogonal code for any linear code $C$. Note that $\hull(C_1) \subseteq \hull_{C_2}(C_1)$ if $C_2 \subseteq C_1$ and $\hull_{C_2}(C_1) \subseteq \hull(C_1)$ if  $C_1 \subseteq C_2$.

The following result presents some basic properties of the relative hull.

\begin{proposition}\label{Prop:basic}
Let $C_i$ be an $[n,k_i]_q$-code with generator matrix $G_i$ for $i=1,2$. The following hold: \newline
\noindent {\rm (i)} $\hull_{C_2}(C_1) = \left\{ {x}G_1 \mid {x} \in \ker (G_2G_1^T) \right\}$, \newline
\noindent {\rm (ii)} $\dim\hull_{C_2}(C_1)=k_1-\rk(G_2G_1^T),$ and \newline
\noindent {\rm (iii)} $k_1-\dim\hull_{C_2}(C_1)=k_2-\dim\hull_{C_1}(C_2).$
\end{proposition}
\begin{proof}
{\rm (i)} $(\subseteq)$
If ${c} \in \hull_{C_2}(C_1)=C_1\cap C_2^{\perp},$ then ${c}={x}G_1$ for some ${x}\in \mathbb{F}_q^{k_1}$ and $G_2{c}^T=~0$. Hence, $G_2G_1^T{x}^T= 0,$ which means that ${x} \in \ker (G_2G_1^T).$
We conclude that ${c} \in \left\{ {x}G_1 \mid {x} \in \ker (G_2G_1^T) \right\}.$

$(\supseteq)$ If ${c} \in \left\{ {x}G_1 \mid {x} \in \ker (G_2G_1^T) \right\}$ then there is ${x} \in \ker (G_2G_1^T)$ such that ${c}={x}G_1$ indicating that ${c}\in C_1$. Furthermore, $G_2{c}^T = G_2G_1^T{x}^T=0$, demonstrating that ${c} \in C_2^\perp$. Thus, ${c}\in C_1\cap C_2^\perp = \hull_{C_2}(C_1).$

{\rm (ii)} The matrix $G_1 \in \mathbb{F}_q^{k_1\times n}$ has rank $k_1$, so it defines the injective transformation $T_{G_1} \colon \mathbb{F}_q^{k_1} \to \mathbb{F}_q^{n}$ given by ${x} \mapsto {x}G_1$. Combining this fact with (i) shows 
$$\begin{array}{lcl}
\dim \hull_{C_2}(C_1) &=& \dim \left\{ {x}G_1 \mid {x} \in \ker (G_2G_1^T) \right\}\\ &=&\dim \left\{ {x} \mid {x} \in \ker (G_2G_1^T) \right\} \\& =& \dim \ker (G_2G_1^T) = k_1-\rk(G_2G_1^T). \end{array}$$

{\rm (iii)} This is a consequence of $\rk(G_2G_1^T)=\rk(G_1G_2^T)$ and (ii).
\end{proof}

A \textit{monomial matrix} is an invertible matrix with rows of weight one. If all nonzero entries of a monomial matrix are ones, it is called a {\it permutation matrix}.
\begin{definition}
Two codes $C$ and $C^\prime$ over $\F_q$ of the same length are {\it monomially equivalent}, or {\it equivalent} for short, if there exists a monomial matrix $M$ such that
$$C^\prime=CM=\{{c}M\ |\ {c}\in C\}.$$
\end{definition}

In fact, according to MacWilliams' theorem, every isometry on $\F_q^n$ with respect to the Hamming metric is given by a monomial matrix~\cite[Theorem 4]{macwilliams1960}. As monomial equivalence preserves the weight distributions, equivalent codes have the same basic parameters: length, dimension, and minimum distance. It is easy to see that the duals of equivalent codes are equivalent. More precisely, $C$ and $C'$ are equivalent with $C^\prime=CM$ if and only if
${C^\prime}^\perp$ and $C^\perp$ are equivalent with ${C^\prime}^\perp=C^\perp P D^{-1}$, where  $M=PD$, $P$ is a permutation matrix, and $D$ is a nonsingular diagonal matrix. 

Given two codes $C_1, C_2 \subseteq \F_q^n$, we aim to find equivalent codes that define a relative hull of dimension that is increased or decreased by one from that of the hull of the original codes and then proceed iteratively. More precisely, we are looking for codes $C_1^\prime$ and $C_2^\prime$ equivalent to $C_1$ and $C_2$, respectively, such that $\dim \hull_{C^\prime_2}(C^\prime_1) = \dim \hull_{C_2}(C_1)+1$ or $\dim \hull_{C^\prime_2}(C^\prime_1)=\dim \hull_{C_2}(C_1)-1$. The following observation shows that modifying only one of the codes is enough to increase or decrease the relative hull dimension. In other words, when we look for codes $C_1^\prime$ and $C_2^\prime$  equivalent to $C_1$ and $C_2$ such that $\dim \hull_{C^\prime_2}(C^\prime_1)=\dim \hull_{C_2}(C_1)+1$ or $\dim \hull_{C^\prime_2}(C^\prime_1)=\dim \hull_{C_2}(C_1)-1$, we can always take $C_2^\prime = C_2$.
\begin{proposition}\label{22.11.01}
If $C_i \subseteq \F_q^n$ is a code and $M_i \in \F_q^{n \times n}$ is a monomial matrix for $i=1,2$, then
$$\dim \hull_{C_2M_2}(C_1M_1)=\dim \hull_{C_2M}(C_1)=\dim \hull_{C_2}(C_1M^T),$$
where $M = M_2M_1^T$.
\end{proposition}
\begin{proof}
Let $G_1$ and $G_2$ be generator matrices for $C_1$ and $C_2$, respectively.
By Proposition~\ref{Prop:basic} (ii),
$$
\dim \hull_{C_2M_2}(C_1M_1)=k_1-\rk(G_2M_2(G_1M_1)^T)=k_1-\rk(G_2MG_1^T)=\dim \hull_{C_2M}(C_1).
$$
 Noting that $G_2MG_1^T=G_2(G_1M)^T$, we also see that
$$
\dim \hull_{C_2M}(C_1)=\dim \hull_{C_2}(C_1M^T),
$$
which proves the assertion.
\end{proof}

\subsection{Hermitian and Galois relative hulls}\label{22.12.25}
In~\cite{fan_zhang}, Fan and Zhang introduced the Galois inner products, a generalization of the Euclidean and Hermitian inner products, and found self-orthogonal codes with respect to the new inner product. The Galois inner products were further studied to build LCD codes \cite{liu_pan} and to get new families of quantum codes with a broader range of parameters (see, for example, \cite{mdsgalois, eaqec_lcd}). This section reviews the Galois inner products and the relative hulls with respect to them. It also demonstrates why, for our purposes, it is sufficient to focus on the classical Euclidean relative hull (rather than these more general Galois relative hulls).

Consider the finite field $\F_q$, where $q=p^m$ for a prime $p$ and a positive integer $m$. 
For any integer $e$ such that $0 \leq e < m$, the {\it $e$-Galois inner} product for ${x}, {y} \in \F_q^n$ is given by 
$${x} \cdot_e {y} = \sum_{i=1}^nx_iy_i^{p^e} \in \F_q.$$
Taking $e=0$ recovers the Euclidean inner product in $\F_{q}^n$. Taking $e=\frac{m}{2}$ when $m$ is even produces the usual Hermitian inner product in $\F_{q}^n$ that is denoted by ${x} \cdot_h {y}$. The {\it $e$-Galois dual} of a code $C \subseteq \F_q^n$ is defined by
$$C^{\perp_e} = \left\{ {x} \in \mathbb{F}_{q}^n \mid {x} \cdot_e {c} = 0, \text{ for all } {c} \in C \right\}.$$
The {\it Hermitian dual} is denoted by $C^{\perp_h}$.
Given two codes $C_1$ and $ C_2$ over $\F_q$, we define the {\it $e$-Galois relative hull} of $C_1$ with respect to $C_2$ as
$$\hull^e_{C_2}(C_1) = C_1\cap C_2^{\perp_e}.$$
We denote the {\it Hermitian relative hull} by $\hull^h_{C_2}(C_1)$. The $e$-Galois relative hulls $\hull^e_{C_1}(C_1)$ and $\hull^h_{C_1}(C_1)$ are denoted respectively by $\hull_e(C_1)$ and $\hull_h(C_1)$.

Given a code $C \subseteq \F_q^n$, consider the code
 $$C^{p^e} =\{(c_1^{p^e},\ldots,c_n^{p^e})\ |\ (c_1,\ldots,c_n)\in C\}.$$ Since the map $\mathbb{F}_{q} \rightarrow \mathbb{F}_{q}: x\mapsto x^{p^e}$ is bijective, we have that if $G=[a_{ij}] \in \F_{q}^{k\times n}$ is a generator matrix of $C$, then $G^{p^e}=[a_{ij}^{p^e}] \in \F_{q}^{k\times n} $ is a generator matrix of $C^{p^e}$. Moreover, 
 $$C^{\perp_e} = (C^{p^e})^\perp.$$
 Thus,
\begin{equation} \label{e_Gal_red}
\hull^e_{C_2}(C_1) = \hull_{C_2^{p^e}}(C_1) \qquad \text{ and } \qquad \hull_e(C)=\hull_{C^{p^e}}(C).\end{equation}
Consequently, to consider the relative hull of a code $C_1$ with respect to $C_2$ and any $e$-Galois inner product, it suffices to consider the relative hull of $C_1$ with respect to $C_2':=C_2^{p^e}$ and the Euclidean inner product.

\rmv{In this work, we are interested, for the most part, in linear isometries. However, if we want to consider non-linear isometries $f$ of $C_2\subseteq\mathbb{F}_q^n$, a desirable feature would be that $f(C_2)$ is a linear code to use the CSS construction. It is known (see \cite{nonlinisom}) that any non-linear isometry that maps subspaces into subspaces is a semilinear isometry, meaning there exists $\lambda\in(\mathbb{F}_q^\ast)^n$, a permutation matrix $P$, and an automorphism $\sigma$ of $\mathbb{F}_q$, such that
$$f(c_1,\ldots,c_n)=(\lambda_1\sigma(c_1),\ldots,\lambda_n\sigma(c_n))P.$$
Thus, considering the $e$-Galois hulls, we will have considered all the isometries of $C_2$ whose image is still a linear code.}

\subsection{Quantum codes}\label{22.12.26}
A series of works in the 1990s showed how a self-orthogonal code or two linear codes subject to a dual-containment constraint give rise to quantum error-correcting codes. Since then, many quantum codes in the literature have relied on the dual of a code. In 2006, Brun, Devetak, and Hsieh~\cite{brun_science} demonstrated that the duality requirement could be removed by using the entanglement, paving the way for any linear code or pair of linear codes to design Entanglement-Assisted Quantum Error-Correcting Codes (EAQECCs). The cost of the pre-shared entanglement can affect the analysis of the performance of a code. Thus, looking for constructions with different required numbers of pairs of maximally entangled qudits is valuable. Moreover, EAQECCs have been used recently for secret sharing \cite{shibatamatsumoto}. Building on the work of Wilde and Brun~\cite{wilde}, Guenda, Jitman, and Gulliver~\cite{guenda} showed that the dimension of the hull of the linear code could capture the necessary entanglement. In this subsection, we review the concepts from the recent work~\cite{GHMR19, GaHeMaRu} that motivate the remainder of this paper.

Recall that the standard notation $[[n,\kappa,\delta;c]]_q$ describes a quantum code $Q$ that is a $q$-ary EAQECC that encodes $\kappa$ logical qudits into $n$ physical qudits with the help of $n - \kappa - c$ ancillas and $c$ pairs of maximally entangled qudits; the code is able to detect any error affecting at most $d-1$ of the physical qudits. If for any error $E$ affecting less than $d$ qudits, we have $v^T E u=0$ for any $v,u\in Q$, we say that $Q$ is pure.

There are several constructions of EAQECCs using linear codes. For example, we have the following two classical constructions using the Euclidean and the Hermitian inner products.

\begin{theorem}[CSS construction, {\cite[Theorem~4]{GHMR19}}]\label{22.10.01}
If $C_i$ is an $[n,k_i]_q$-code for $i=1,2$, then there exists an $[[n,\kappa,\delta;c]]_q$-quantum code $Q$ with
        $$c=k_1-\dim\hull_{C_2}(C_1),\qquad \kappa=n-k_1-k_2+c,\quad $$
\begin{equation*}
        \text{and} \qquad \delta = \begin{cases}
                        \min\left\{d(C_1^\perp),\ d(C_2^\perp)\right\} &\text{if $C_1^\perp \subseteq C_2$}  \\
                        \min\left\{\wt \left(C_1^\perp\setminus \hull_{C_1}(C_2)\right),
                        \wt \left(C_2^\perp\setminus \hull_{C_2}(C_1)\right)\right\} &\text{otherwise}.
                    \end{cases}
\end{equation*}  
Moreover, if $\delta=\min\{d(C_1^\perp),d(C_2^\perp)\}$, then $Q$ is pure. 
\end{theorem}

\begin{theorem}[Hermitian construction, {\cite[Theorem~3]{GHMR19}}]\label{22.10.02}
    If $C$ is an $[n,k]_{q^2}$-code, then there exists an $[[n,\kappa,\delta;c]]_q$-quantum code $Q$ with
    $$c=k-\dim\hull_h(C),\qquad \kappa=n-2k+c,\quad $$
\begin{equation*}
        \text{and} \qquad \delta = \begin{cases}
                        d(C^{\perp_h}) &\text{if $C^{\perp_h} \subseteq C$}  \\
                        \min \left\{ \wt(C^{\perp_h} \setminus \hull_h(C)) \right\} &\text{otherwise}.
                    \end{cases}
\end{equation*}
Moreover, if $\delta=d(C^{\perp_h})$, then $Q$ is {pure}. 
\end{theorem}

The following Singleton-type bound holds for the CSS and Hermitian constructions.

\begin{theorem}[Singleton-type bound \cite{Grassl22}]
If $Q$ is an $[[n,\kappa,\delta; c]]_q$-quantum code obtained via the CSS or the Hermitian construction, then
$$2\delta+\kappa\leq n+c+2.$$
\end{theorem}

\begin{remark}
Let $C_i$ be an $[n,k_i]_q$-code with generator matrix $G_i$ for $i=1,2$. Note that Proposition~\ref{Prop:basic} (ii) implies that if $Q$ is a quantum code constructed via the CSS construction using the codes $C_1$ and $C_2$, then the parameter $c$, the required number of pairs of maximally entangled quantum states, can be seen in terms of the generator matrices:
$$c=\rk(G_2G_1^T)=\rk(G_1G_2^T).$$

This implies that swapping the role of $C_1$ and $C_2$ does not affect the parameters of the resulting quantum code.
\end{remark}

\section{Reducing the relative hull}
\label{reducerelativehull}
Let $C_i$ be an $[n,k_i]_q$-code for $i=1,2$. This section aims to repeatedly reduce the relative hull dimension $\hull_{C_2}(C_1)$ by one using equivalent codes. We use the phrase reduce the (dimension of the) relative hull to mean to determine equivalent codes that define a relative hull of dimension less than that of the original codes. According to Proposition~\ref{22.11.01}, we only need to find an equivalent code for one of the linear codes. Thus, we seek a code $C_2^\prime$ equivalent to $C_2$ such that $\dim \hull_{C_2^\prime}(C_1)=\dim \hull_{C_2}(C_1)-1$.

For any ${\lambda}=(\lambda_1,\ldots,\lambda_n)\in(\mathbb{F}_q^\ast)^n$, we define the diagonal matrix $D_{\lambda}=\diag(\lambda_1,\ldots,\lambda_n)$. Let $C \subseteq \F_q^n$ be a code and $S_n$ the symmetric group on $n$ symbols. If $\sigma\in S_n$, the image of $C$ obtained by permuting the entries of every codeword according to $\sigma$ is denoted by $C^\sigma$. The permutation matrix associated with $\sigma$ is denoted by $P_\sigma$.

\begin{remark}\label{22.11.04}
Note that $C^\sigma=\left\{{c}P_\sigma \ | \ {c}\in C\right\}.$ Any monomial matrix $M$ is of the form $M=D_{\lambda} P_\sigma$, for some ${\lambda} \in(\mathbb{F}_q^\ast)^n$ and some permutation $\sigma\in S_n.$ Thus, any code $C^\prime$ monomially equivalent to $C$ is of the form $C^\prime = CD_{\lambda} P_\sigma$.
\end{remark} 
When equivalent codes reduce the dimension of the relative hull, the following lemma specifies how much the dimension can be reduced.

\begin{lemma}\label{22.11.02}
Let $C_i$ be an $[n,k_i]_q$-code for $i=1,2$. If $C_2^\prime$ is equivalent to $C_2$, then
$$\dim \hull_{C_2^\prime}(C_1)  \geq \max \{ 0, k_1-k_2 \}.$$   \end{lemma}
\begin{proof}
By Remark~\ref{22.11.04}, there exists a monomial matrix $M$ such that
$C_2' = C_2M$. Let $G_1$ and $G_2$ be generator matrices of $C_1$ and $C_2$, respectively.
By Proposition~\ref{Prop:basic} (ii), $\dim \hull_{C'_2}(C_1) = k_1 - \rk(G_2 M G_1^T)$. The result follows as $G_2 M G_1^T$ is a $k_2 \times k_1$ matrix.
\end{proof}

Lemma \ref{22.11.02} indicates that the dimension of the relative hull of a code $C_1$ with respect to $C_2$ can be reduced (at most) to the difference in dimensions of the two codes, in the case that the difference is nonnegative, by replacing $C_1$ with an equivalent code.

One of the main results of this section proves that one can repeatedly decrease the dimension of the relative hull by one until it equals the lower bound given by Lemma~\ref{22.11.02}. 

Recall that the tensor product of matrices $A=[a_{ij}] \in \F_{q}^{r\times n}$ and $B \in \F_{q}^{m_1\times m_2}$ is the matrix that is expressed in block form as
\[ A\otimes B = \left(\begin{array}{cccc} a_{11}B & a_{12}B & \cdots & a_{1n}B
\\ a_{21}B & a_{22}B & \cdots & a_{2n}B \\
\vdots & \vdots &  & \vdots \\
a_{r1}B & a_{r2}B & \cdots & a_{rn}B \\
\end{array}\right) \in {\F_{q}}\!^{rm_1\times nm_2}.
\]

For any two matrices $A \in \mathbb{F}_q^{r\times n}$ and $B \in \mathbb{F}_q^{n\times s}$, their (usual) product can be seen as $AB=\sum_{i=1}^n {\rm Col}_i(A)\otimes {\rm Row}_i(B)$, where we use ${\rm Col}_i(A)$ (resp. ${\rm Row}_i(A)$) to denote the $i$-th column (resp. row) of $A$. Thus, for ${\lambda} \in (\mathbb{F}_q^\ast)^n$, we have
\begin{equation}\label{22.10.27}
AD_{\lambda}B=\sum_{i=1}^n\lambda_i{\rm Col}_i(A)\otimes {\rm Row}_i(B)=AB+\sum_{i=1}^n(\lambda_i-1){\rm Col}_i(A)\otimes {\rm Row}_i(B).
\end{equation}
If $P=P_{(i j)}$ is the permutation matrix that interchanges rows $i$ and $j$, then
\begin{equation}\label{22.10.28}
APB=AB+\left({\rm Col}_j(A)-{\rm Col}_i(A)\right)\otimes \left({\rm Row}_i(B)-{\rm Row}_j(B) \right).
\end{equation}

Now, we will successively decrease the dimension of a relative hull, say $\hull_{C_2}(C_1)$, by one via equivalent codes. 

\begin{theorem}\label{22.10.29}
Let $C_i$ be an $[n,k_i]_q$-code for $i=1,2$ with $q > 2$. For any $\ell$ with $\max\{0,k_1-k_2\} \leq \ell \leq \dim\hull_{C_2}(C_1)$, there exists a code $C_{2,\ell}$ equivalent to $C_2$ such that $$\dim\hull_{C_{2,\ell}}(C_1)=\ell.$$ Therefore, the dimension of the relative hull of $C_1$ with respect to $C_2$  can be repeatedly decreased by one until it is equal to $\max\{0,k_1-k_2\}$ by replacing $C_2$ with an equivalent code.
\end{theorem}

\begin{proof}
Define $\ell_1=\dim\hull_{C_2}(C_1)$ and $\ell_2=\dim\hull_{C_1}(C_2)$. We may assume that $\hull_{C_1}(C_2)$ is given by a generator matrix $[I_{\ell_2}\ A_2]$
where $I_{\ell_2}$ is an identity matrix of size $\ell_2$, since we seek a code equivalent to $C_2$. 
Extend $[I_{\ell_2}\ A_2]$ to a generator matrix 
$$G_2=\begin{pmatrix}
I_{\ell_2}&A_2\\0&B_2\end{pmatrix}$$ of $C_2$.
Similarly, let $[A_1\  B_1]$ be a generator matrix of $\hull_{C_2}(C_1)$, where $A_1$ is of size $\ell_1\times \ell_2$, and
$$G_1=\begin{pmatrix}
A_1&B_1\\
D_1&E_1\end{pmatrix}$$
is a generator matrix of $C_1$. Observe that $[I_{\ell_2}\ A_2]G_1^T=0$ and $[A_1\ B_1] G_2^T=0$, since the first matrix in each product has rows in the dual of the code generated by the second term of each product, then
$$G_2G_1^T=\begin{pmatrix}
A_1^T+A_2B_1^T& D_1^T+A_2E_1^T\\
B_2B_1^T& B_2E_1^T
\end{pmatrix}=\begin{pmatrix}
0&0\\
0&B_2E_1^T
\end{pmatrix},$$
where $B_2E_1^T$ is a $(k_2-\ell_2)\times (k_1-\ell_1)$ matrix.
By Proposition~\ref{Prop:basic} (iii), $k_2-\ell_2=k_1-\ell_1$, so $B_2E_1^T$ is a square matrix. This, together with Proposition~\ref{Prop:basic}(ii), implies that $B_2E_1^T$ has full rank.
The goal is to increase the rank of $G_2G_1^T$, meaning to determine a code equivalent to $C_2$ with generator matrix $G_2'$ so that 
$\rk \left(G_2'G_1^T \right) > \rk \left( G_2G_1^T \right)$.

Case 1: Assume $A_1\neq 0.$ Then there is $1\leq j\leq \ell_2$ such that ${\rm Row}_j(G_1^T) \neq 0$. Set ${\lambda} = (1,\ldots,1,\lambda_j,1,\ldots,1) \in (\mathbb{F}_q^\ast)^n$ to be the vector with all entries equal to $1$ except in position $j$ where the entry is $\lambda_j \neq 1$.
By Eq.~(\ref{22.10.27}), we have
$$
\begin{array}{lcl}
G_2\ D_{\lambda} G_1^T&=&
G_2 G_1^T +
(\lambda_j-1) {\rm Col}_j(G_2) \otimes {\rm Row}_j(G_1^T) \\ \ \\
&=&\begin{pmatrix}
(\lambda_j-1)\bm{e}_j^T \otimes {\rm Row}_j(A_1^T)&(\lambda_j-1)\bm{e}_j^T\otimes {\rm Row}_j(C_1^T)\\
0&B_2E_1^T\end{pmatrix}.
\end{array}$$
Observe that $\rk \left( G_2\ D_{\lambda} G_1^T \right) =k_2-\ell_2+1$, because $\lambda_j \neq 0, 1.$

Case 2: Assume $A_1 = 0.$ In this case,
$G_1=\begin{pmatrix}
0&B_1\\
D_1&E_1\end{pmatrix}.$
Recall that $B_1 \in \F_q^{\ell_1 \times (n-\ell_2)}$ has full rank. After row operations, we may consider that there are $\ell_1$ integers $1\leq i_1 < \ldots < i_{\ell_1}\leq n-\ell_2$ such that ${\rm Col}_{i_j}(B_1)=\bm{e}^\prime_j $ and ${\rm Col}_{i_j}(E_1)=0$ for $1\leq j\leq \ell_1$. 

Subcase (i): Assume that for some $1\leq j\leq \ell_1$, ${\rm Col}_{i_j}(A_2)\neq 0$. Let $\nu=\ell_2+i_j$. For an element $\lambda_\nu \in \mathbb{F}_q^\ast$ such that $\lambda_\nu \neq 1$, define ${\lambda}=(1, \dots, 1, \lambda_\nu, 1, \dots, 1)  \in (\mathbb{F}_q^\ast)^n$ as the vector with all entries equal to $1$ except in position $\nu$ where the entry is $\lambda_\nu$. Then the matrix
$$G_2D_{\lambda}G_1^T=\begin{pmatrix}
(\lambda_\nu-1){\rm Col}_{i_j}(A_2)\otimes \bm{e}^\prime_j &0\\
(\lambda_\nu-1){\rm Col}_{i_j}(B_2)\otimes \bm{e}^\prime_j&B_2E_1^T
\end{pmatrix}$$
has rank $k_2-\ell_2+1.$

Subcase (ii): Assume that ${\rm Col}_{i_j}(A_2)=0$ for all $1\leq j\leq\ell_1$. Let $P$ be the permutation matrix that interchanges rows $1$ and $\ell_2+i_1$. By Eq.~(\ref{22.10.28}),
$$G_2PG_1^T=G_2G_1^T+\begin{pmatrix}
-\bm{e}_1^T \otimes -\bm{e}^\prime_{1}& -\bm{e}_1^T \otimes {\rm Row}_1(D_1^T)\\
{\rm Col}_{i_1}(B_2)\otimes -\bm{e}^\prime_{1}&{\rm Col}_{i_1}(B_2)\otimes {\rm Row}_{1}(D_1^T)
\end{pmatrix}.$$
Since the row space of the second term is generated by the row $(-\bm{e}^\prime_{1}, {\rm Row}_1(D_1^T))$, then the matrix $G_2PG_1^T$ has rank $k_2-\ell_2+1$.

Take $G_2'=G_2 P$. Then $C_2$ is equivalent to the code $C_2'$ with generator matrix $G_2'$. Moreover, in any case, $$\rk \left(G_2'G_1^T \right) = \rk \left( G_2G_1^T \right)+1.$$
{According to Proposition~\ref{Prop:basic}(ii), 
$$\dim\hull_{C'_2}(C_1) =k_1-\rk\left(G_2'G_1^T \right)= k_1-\left(\rk \left( G_2G_1^T \right)+1 \right)=\dim\hull_{C_2}(C_1)-1,$$}meaning
we have decreased the dimension $\dim\hull_{C_2}(C_1)$ of the relative hull  by one. We can continue this process until the rank of the matrix $G_2PG_1^T$ is $k_2$, which means $\dim\hull_{C'_2}(C_1)=\max\{0,k_1-k_2\}$.
\end{proof}

Algorithm \ref{alg:1} captures the procedure written in the proof of Theorem \ref{22.10.29}. The input and the output are given in terms of the generator matrices of the pair of codes. To simplify this algorithm, we first use Algorithm \ref{alg:0} so that the generator matrices are of the appropriate form.

\begin{algorithm}
\caption{Systematic-like form for the generator matrices}\label{alg:0}
\KwData{$G_1\in\mathbb{F}_q^{k_1\times n}$, $G_2\in\mathbb{F}_q^{k_2\times n}$ full-rank matrices.}
\KwResult{$G'_1\in\mathbb{F}_q^{k_1\times n}$, $G'_2\in\mathbb{F}_q^{k_2\times n}$}
$(k_1,k_2)\gets (\rk{G_1},\rk{G_2})$

$(\ell_1,\ell_2)\gets (k_1-\rk(G_2G_1^T),k_2-\rk(G_2G_1^T))$

For $i=1,2$, pick $M_i\in\mathbb{F}_q^{k_i\times k_i}$ be a non-singular matrix such that the first $\ell_i$ rows are in $\ker(G_{1+(i\% 2)}G_{i}^T)$.

$(G_1,G_2)\gets(M_1G_1, M_2G_2)$

Pick $M_3$ a non-singular matrix, $P$ a permutation matrix such that $(M_3)_{i,j}=0$ if $i\leq\ell_2$ and $j\geq \ell_2$, and $M_3G_2P=\begin{pmatrix}
I_{\ell_2}& A_2\\ 0& B_2\end{pmatrix}$.

Let $M_4$ be a non-singular matrix such that $(M_4)_{ij}=0$ if $i\leq \ell_1$ and $j\geq \ell_1$ and $M_4G_1$ is in row-reduced-echelon form.

$G'_1\gets M_4G_1$

$G'_2\gets M_3G_2P$
\end{algorithm}

\begin{algorithm}
\caption{Reducing the hull of two codes}\label{alg:1}
\LinesNumbered
\KwData{$G_1\in\mathbb{F}_q^{k_1\times n}$, $G_2\in\mathbb{F}_q^{k_2\times n}$ full-rank matrices.}

\KwResult{$G'_2$ a full-rank matrix with $\rk(G_1(G'_2)^T)=\rk(G_1G_2^T)+1$.}

Replace $(G_1,G_2)$ with the result of Algorithm \ref{alg:0}.

\eIf{$\left[(G_1)_{ij}\right]_{i,j=1}^{\ell_1}\neq 0$}{$j\gets \min\{h\in[\ell_1]\ :\ \exists i\in[\ell_1],\ (G_1)_{ij}\neq 0\}$

Take $\lambda_j\in\mathbb{F}_q\setminus\{0,1\}$.

$\lambda\gets \lambda_j e_j + \sum_{i \in [n] \setminus \{ j \} } e_i$

$G'_2\gets G_2D_{\lambda}$

}
{
\eIf{$\exists j\in[n]$ such that $wt(\mathrm{Col}_j(G_1))=1$ and $\mathrm{Col}_j(G_2)\neq 0$}{
Take $\lambda_j\in\mathbb{F}_q\setminus\{0,1\}$.

$\lambda\gets \lambda_j e_j + \sum_{i \in [n] \setminus \{ j \} } e_i $

$G'_2\gets G_2D_{\lambda}$

}{Take $j\in[n]$ such that $\mathrm{Col}_j(G_1)=e_1$. 

Take $P'$, the permutation matrix that permutes rows $1$ and $j$. 

$G'_2\gets G_2P'$
}}
\end{algorithm}

We now give some examples to illustrate how the proof of Theorem~\ref{22.10.29} constructs equivalent codes that reduce the relative hull, using ~\cite{macaulay2, magma, Mac2} to make the computations.

\begin{example}
Let $a$ be a primitive element of $\mathbb{F}_9$, with $a^2-a-1=0$, and $C_1$ and $C_2$ the codes over $\mathbb{F}_9$ generated respectively by
	$$G_1=\begin{pmatrix}
      1&0&0&0&0&1&a\\
      0&1&0&0&-a-1&-a-1&a\\
      0&0&1&0&a+1&a+1&a+1\\
      0&0&0&1&0&0&0
      \end{pmatrix}$$
and
     $$G_2=\begin{pmatrix}
      1&0&0&0&1&-1&0\\
      0&1&0&0&1&-a-1&a\\
      0&0&1&0&a-1&-a-1&a\\
      0&0&0&1&0&0&0
      \end{pmatrix}.$$
           The subspaces $\hull_{C_2}(C_1)$ and $\hull_{C_1}(C_2)$ are generated by the first three rows of $G_1$ and $G_2$, respectively. This example corresponds to the proof of Theorem~\ref{22.10.29}, Case 1. We only need to choose  $\lambda$ with entries different from $1$ since the first three entries of the main diagonal are non-zero.
     
      For $0\leq\ell\leq 3$, let $\lambda^{(\ell)}\in\mathbb{F}_9^7$ be the vector such that $\left(\lambda^{(\ell)}\right)_i = a$ for $1\leq i\leq 3-\ell$ and $\left(\lambda^{(\ell)}\right)_i = 1$ for $i\geq 3-\ell$. Let $C_{2,\ell}$ be the code generated by $G_2D_{\lambda^{(\ell)}}$. We have 
     $$G_2D_{\lambda^{(\ell)}}G_1^T=\begin{pmatrix}
      \lambda_{1}^{(\ell)}-1&0&0&0\\
      0&\lambda_{2}^{(\ell)}-1&0&0\\
      0&0&\lambda_{3}^{(\ell)}-1&0\\
      0&0&0&1
      \end{pmatrix}.$$
Therefore, $\rk (G_2D_{\lambda^{(\ell)}}G_1^T)=4-\ell$ and thus $\dim\hull_{C_{2,\ell}}(C_1)=\ell$.
\end{example}

\begin{example}\rm
Let $a$ be a primitive element of $\mathbb{F}_9$, with $a^2-a-1=0$, and $C_1$ and $C_2$ the codes over $\mathbb{F}_9$ generated respectively by
 $$G_1=\begin{pmatrix}
       0&0&1&-1&0&0\\
       0&0&0&0&1&-1\\
       -a&0&1&0&0&0\\
       0&-a-1&0&0&1&0
       \end{pmatrix}
\textnormal{ and }
     G_2=\begin{pmatrix}
        1&0&a&a&0&0\\
       0&1&0&0&a+1&a+1\\
       0&0&1&1&0&0\\
       0&0&0&0&1&1
       \end{pmatrix}.$$
The relative hulls are generated by the first two columns of each matrix. As $G_1$ has its principal minor of size 2 equal to zero, this example corresponds to the proof of Theorem~\ref{22.10.29}, Case 2. We can use the first two entries of the last four columns of $G_2$ to modify the hull size (Subcase (i)) because they are non-zero. Let $\lambda^{(1)}\in\mathbb{F}_9^6$ such that $\lambda^{(1)}_i = 1$ for $i\neq 6$ and $\lambda^{(1)}_6 = a$. Let $C_{2,1}$ be the code generated by $G_2D_{\lambda^{(1)}}$. The matrix
  $$G_2D_{\lambda^{(1)}}G_1^T=\begin{pmatrix}  0&0&0&0\\
       0&-a&0&0\\
       0&0&1&0\\
       0&-a+1&0&1\end{pmatrix}$$
has rank $3$ and $\dim\hull_{C_{2,1}}(C_1)=1$. We can check that the last three rows of $G_2$ do not belong to $\hull_{C_{1}}(C_{2,1})$, so we are still in Case 2, Subcase (i) of the proof of Theorem~\ref{22.10.29}. Let $\lambda^{(2)}\in\mathbb{F}_9^6$ such that $\lambda^{(2)}_i = 1$ for $i\neq 4$ and $\lambda^{(2)}_4 = a$. Let $C_{2,2}$ be the code generated by $G_2D_{\lambda^{(1)}}D_{\lambda^{(2)}}$. The matrix
       $$G_2D_{\lambda^{(1)}}D_{\lambda^{(2)}}G_1^T=\begin{pmatrix}
       -1&0&0&0\\
       0&-a&0&0\\
       -a+1&0&1&0\\
       0&-a+1&0&1
       \end{pmatrix}$$
has rank $4$ and $\dim\hull_{C_{2,2}}(C_1)=0$.
\end{example}

\begin{example}
Let $a$ be a primitive element of $\mathbb{F}_9$, with $a^2-a-1=0$, and $C_1$ and $C_2$ the codes over $\mathbb{F}_9$ generated respectively by
	$$G_1=\begin{pmatrix}
			0&0&-a&-a&1&0\\
       0&0&-a&-a&0&1\\
       0&0&1&0&0&0\\
       0&0&0&1&0&0
      \end{pmatrix}
\textnormal{ and }
      G_2=\begin{pmatrix}
      		1&0&0&0&0&0\\
       0&1&0&0&0&0\\
       0&0&1&0&a&a\\
       0&0&0&1&a&a
       \end{pmatrix}.$$
The relative hulls are generated by the first two rows of each matrix. The principal minor of size 2 of $G_1$ is 0, so this example corresponds to the proof of Theorem~\ref{22.10.29}, Case 1. Since the $(G_2)_{i,j}=0$ for $i=1,2$ and $3\leq j\leq 6$, we are in the Subcase (ii). We need to perform some column permutations to $G_2$ to get an equivalent code with a smaller relative hull than $C_2$.

Let $P_1$ be the permutation matrix that permutes columns $5$ and $1$, and let $C_{2,1}$ be the code generated by $G_2P_1$. The matrix
$$G_2P_1G_1^T=\begin{pmatrix}
       1&0&0&0\\
       0&0&0&0\\
       -a&0&1&0\\
       -a&0&0&1\end{pmatrix}$$
has rank 3, therefore $\dim \hull_{C_{2,1}}(C_1)=1$.

Let $P_2$ be the permutation matrix that permutes columns $2$ and $6$, and let $C_{2,0}$ be the code generated by $G_2P_1P_2$. The matrix
       $$G_2P_1P_2G_1^T=\begin{pmatrix}1&0&0&0\\
       0&1&0&0\\
       -a&-a&1&0\\
       -a&-a&0&1\end{pmatrix}$$
has rank 4 and thus, $\dim\hull_{C_{2,0}}(C_1)=0$. 
\end{example}
Let $C_1$ and $ C_2$ be two codes over $\F_q$ with $q = p^m > 2$, and let $e$ be an integer such that $0 \leq e < m$. Applying Theorem \ref{22.10.29} to the relative hull of $C_1$ with respect to $C_2^{p^e}$, we obtain a similar result for the $e$-Galois hull of $C_1$ with respect to $C_2$. This consequence is captured in the next statement.

\begin{corollary}
Let $C_i$ be an $[n,k_i]_q$-code for $i=1,2$ with $q=p^m>2$. Take $e$ such that $0 \leq e < m$. For any $\ell$ with $\max\{0,k_1-k_2\} \leq \ell \leq \dim\hull^e_{C_2}(C_1)$, there is a code $C_{2,\ell}$ equivalent to $C_2$ such that $$\dim\hull^e_{C_{2,\ell}}(C_1)=\ell.$$ Therefore, the dimension of the $e$-Galois relative hull of $C_1$ with respect to $C_2$  can be repeatedly decreased by one until it is equal to $\max\{0,k_1-k_2\}$ by replacing $C_2$ with an equivalent code.
\end{corollary}

\begin{proof} This statement follows immediately from Theorem \ref{22.10.29} and Eq.~(\ref{e_Gal_red}).
\end{proof}

Let $C_i$ be an $[n,k_i]_q$-code for $i=1,2$. If $c_1=\left( c_{11}, c_{12} , \dots, c_{1n}  \right) \in C_1$ and $c_2=\left(  c_{21},  c_{22}, \dots,  c_{2n} \right) \in C_2$, then their {\it Schur product} is defined by
$$c_1\star c_2= \left( c_{11} c_{21}, c_{12} c_{22}, \dots, c_{1n} c_{2n} \right) \in \F_q^n.$$
The Schur product of the codes $C_1$ and $C_2$, denoted by $C_1\star C_2$, is defined as the $\F_q$-vector space spanned by the set
$\left\{c_1\star c_2 \ | \ c_1 \in C_1, c_2\in C_2 \right\}.$ For an element $\lambda$ in $\mathbb{F}_q^n$, $\lambda \star C_2$ denotes the $\F_q$-vector space spanned by the set $\left\{\lambda \star c_2 \ | \ c_2\in C_2 \right\}$.

\begin{proposition}\label{egal}
Let $C \subseteq \mathbb{F}_q^n$ be a code with $q = p^m > 2$.
Take $e$ such that $0 \leq e < m$ and define $\ell=\dim\hull_e(C)$. If there exists $x \in \F_q^\ast$ such that $x^{p^e+1}\neq 1$, then  $\dim\hull_e(\lambda\star C)=\ell-1$ for some $\lambda \in \left(\F_q^\ast\right)^n$.
\end{proposition}

\begin{proof}
	Let $G$ be a generator matrix of $C$. As $C^{\perp_e}=(C^{p^e})^\perp$,
	$$\dim\hull_e(\lambda\star C)=\dim C-\rk(GD_{\lambda^{p^e+1}}(G^{p^e})^T),$$
	where $(G^{p^e})_{ij}=(G_{ij})^{p^e}$. The proof of Theorem~\ref{22.10.29} guarantees that we can reduce the rank of this matrix as long as there exists $x\in\mathbb{F}_q$ with $x^{p^e+1}\neq 1$.
\end{proof}

As a corollary, we can prove some of the significant results that were initially proved by Carlet, Mesnager, Tang, Qi, and Pellikaan (existence of LCD codes for the case of the Euclidean and the Hermitian inner product ~\cite{lcd}) and Luo, Ezerman, Grassl, and Ling (the step-wise reduction of the dimension of the Hermitian hull~\cite{Grassl22}).

\begin{corollary}\label{23.06.24}
Let $C \subseteq \F_q^n$ be a linear code. The following hold:
	\begin{enumerate}
		\item If $q>3$ and $0\leq\ell\leq\dim\hull(C)$, then there is a code $C_\ell$ equivalent to $C$ such that $\hull(C_\ell)=\ell$.
		
		\item If $q> 4$ is a square and $0\leq\ell\leq\dim\hull_h(C)$, then there is a code $C_\ell$ equivalent to $C$ such that $\hull_h(C_\ell)=\ell$. 
	\end{enumerate}
\end{corollary}

\begin{proof}
The Euclidean hull is the $e$-Galois hull with $e=0$. Thus, it is enough to guarantee that $x^2-1\neq 0$ for some $x \in \F_q^\ast$, which happens if $q>3$.

The Hermitian hull is also an $e$-Galois hull where $e$ satisfies $p^e=\sqrt{q}$ and $p$ is the characteristic of the field. By Proposition \ref{egal}, we can reduce the hull using an equivalent code as long as there is $x \in \F_q^\ast$ such that $x^{\sqrt{q}+1}\neq 1$. Note that as $q>4$, $\sqrt{q}+1 < q-1$. Thus, not all the elements of $\F_q^\ast$ can be roots of the polynomial $f(t)=t^{\sqrt{q}+1}-1$, meaning that there is $x \in \F_q^\ast$ such that $x^{\sqrt{q}+1} \neq 1$. 
Another way to see this is by noticing that $x^{\sqrt{q}+1}$ is the norm of $x$ with respect to the extension $\F_q/\F_{\sqrt{q}}$. As the norm is surjective, there are non-zero elements with a norm different from 1 when $q>4$.

\end{proof}

\begin{remark}
If we only consider monomial matrices of the form $M=D_{\lambda}$ in the definition of equivalent codes, meaning no permutations of coordinates are allowed, then it may be impossible to reduce $\dim \hull_{C_2}(C_1)$ to $\max\{0,k_1-k_2\}$. The following example illustrates this fact.
\end{remark}

\begin{example}
Let $C_1$ and $C_2$ be the codes over $\F_q$ generated respectively by
$$G_1= \left( 1\ 1\ 0\ 0 \right) \qquad \text{ and } \qquad G_2= \left( 0\ 0\ 1\ 1 \right).$$
Note that $\max\{0,k_1-k_2\}=0$ and that $G_1D_{\lambda}G_2^T=0$ for any ${\lambda}\in(\mathbb{F}_q^\ast)^n$. Hence, $\dim \hull_{C_2D_{\lambda}}(C_1)=1$ for any $\lambda\in(\mathbb{F}_q^\ast) ^n$.

To get the minimum possible hull, we need permutations. If $P$ is the permutation matrix that interchanges the first and the fourth column, then  $G_1P^TG_2^T=I_1$ and thus $\mathrm{Hull}_{C_2P}(C_1)=0$. 
\end{example}

\section{Increasing the relative hull} \label{increaserelativehull}
Let $C_i$ be an $[n,k_i]_q$-code for $i=1,2$. In this section, we give conditions that allow us to find equivalent codes that successively increase the dimension of the relative hull of  $C_1$ with respect to $C_2$  by one. As in Section~\ref{reducerelativehull}, according to Proposition~\ref{22.11.01}, we only need to show that an equivalent code exists for one of the linear codes. Hence, we aim to determine when it is possible to find {a code $C_1^\prime$ equivalent to $C_1$} such that $\dim \hull_{C_2}(C^\prime_1)=\dim \hull_{C_2}(C_1)+1$.

The following lemma gives an upper bound on the increased dimension of the relative hull. However, as we will see, it is only possible sometimes to increase the dimension of the relative hull using equivalent codes.
\begin{lemma}\label{22.11.03}
Let $C_i$ be an $[n,k_i]_q$-code for $i=1,2$. If $C_1^\prime$ is equivalent to $C_1$, then 
\[
\dim \hull_{C_2}(C_1^\prime) \leq  \min \{ k_1,  n-k_2 \}.
\]
\end{lemma}
\begin{proof} This is clear by the definition of $\hull_{C_2}(C_1')$.
\end{proof}


By Theorem~\ref{22.10.29}, we can decrease the relative hull dimension by increasing the rank of the matrix $G_1G_2^T$. To increase the relative hull dimension instead, we could try to mimic this idea by decreasing the rank of the matrix $G_1G_2^T$ until it is equal to $0$. Unfortunately, the following example shows that reducing the rank of this matrix $G_1G_2^T$ is not always possible.

\begin{example}
Let $C_1$ and $C_2$ be the codes over $\F_q$ generated respectively by
$$G_1=\begin{pmatrix} 1&0&-1&0\\ 0&1&0&-1\end{pmatrix} \qquad \text{ and } \qquad G_2=\begin{pmatrix}1&0&1&1\\ 0&1&0&0\end{pmatrix}.$$
For any permutation matrix $P$ and for any $\lambda\in(\mathbb{F}_q^\ast)^4$, the second column of $G_1D_\lambda PG_2^T$ is either $\pm(\lambda_2\ 0)^T$ or $\pm(0\ \lambda_2)^T$. Thus, the rank of $G_1D_\lambda PG_2^T$ is at least 1. 
\end{example}

We can relate the maximum dimension of the hull under isometries of the form $D_\lambda$ with the dual of the Schur product of the codes.

\begin{proposition}\label{Prop:maxwt}
If $C_i$ is an $[n,k_i]_q$-code for $i=1,2$, then
\[ \max\{\dim \hull_{C_2}(C_1D_{\lambda}) \ |\ {\lambda} \in (\mathbb{F}_q^\ast)^n\} \geq \max\wt\left((C_1\star C_2)^\perp\right)-n+\min\{k_1,k_2\}.\]
\end{proposition}
\begin{proof}
Let $G_1$ and $G_2$ be generator matrices of $C_1$ and $C_2$, respectively.
According to Proposition \ref{Prop:basic}~(ii), we need to show that
$$\min\left\{\rk(G_1D_{\lambda} G_2^T) \ | \ {\lambda}\in(\mathbb{F}_q^\ast)^n \right\}\leq n-\max\wt\left\{(C_1\star C_2)^\perp\right\}.$$
Suppose $\max\mathrm{wt}\left((C_1\star C_2)^\perp\right)=n-\ell$, and take ${\gamma} \in (C_1\star C_2)^\perp$ with $\wt({\gamma} ) =n-\ell$. If $\ell\geq \min\{k_1,k_2\}$, the result follows as $\rk(G_1D_{\lambda} G_2^T)\leq\min\{k_1,k_2\}$ for any $\lambda\in(\mathbb{F}_q^\ast)^n$. Assume that $\ell<\min\{k_1,k_2\}$. Without loss of generality, we can assume that the first $\ell$ entries of ${\gamma}$ are equal to zero. Define ${\lambda} = (1, \dots, 1, \gamma_{\ell +1}, \dots, \gamma_n)$. Then
    $$G_1D_{\lambda} G_2^T=\left(\sum_{h=1}^\ell a_{ih}b_{jh}\right)_{i,j=1}^{k_1,k_2}=G_1\begin{pmatrix} I_\ell&0\\ 0&0\end{pmatrix}G_2^T.$$
Since $\ell<\min\{k_1,k_2\}$, the rank of this product is at most $\ell$, and we have the conclusion.
\end{proof}

In the case where $C_1=C_2$, the code $(C_1\star C_2)^\perp$ was used in \cite{rains} to find self-orthogonal truncations of $C_1$.

It is evident that the bound given by Proposition~\ref{Prop:maxwt} is sharp for codes $C_1$ and $C_2$ such that there is an equivalent code $C'$ to $C_1$ with $C_2^\perp\subseteq C'$. The following example shows that the bound may be sharp even when such an equivalent code does not exist.

\begin{example}\rm
Take $G_1 = G_2 = \begin{pmatrix} 1&0&0\\ 0&1&\beta\end{pmatrix} \in \mathbb{F}_q^{2 \times 3}$ with $\beta\neq 0$.
For any ${\lambda}=(\lambda_1, \ \lambda_2, \ \lambda_3)\in (\F_q^\ast)^3$, we have
$$G_1D_{\lambda}G_2^T=
\begin{pmatrix}
\lambda_1&0\\ 0&\lambda_2+\beta^2\lambda_3
\end{pmatrix}.$$
Then, $\rk \left( G_1D_{\lambda}G_2^T \right) =1$ when $\lambda_2=-\beta^2\lambda_3$; otherwise, $\rk \left( G_1D_{\lambda}G_2^T \right) =2$. Since $1$ is the smallest rank achievable for any ${\lambda}$, the maximum rank of the relative hull is $2-1=1$. 

On the other hand, if $C$ is the code generated by $G_1$, then a generator matrix for the code $C\star C$ is 
$\begin{pmatrix} 1&0&0\\ 0&1&\beta^2\end{pmatrix}$.
It is clear that $(C\star C)^\perp=\langle (0,-\beta^2,1)\rangle$. Then
$$\max\mathrm{wt}((C\star C)^\perp)-n+k_1=2-3+2=1,$$ demonstrating that equality is achievable in Proposition \ref{Prop:maxwt}.
\end{example}

The bound of Proposition~\ref{Prop:maxwt} is an upper bound for the dimension of the relative hull.

\begin{proposition}\label{23.06.13}
If $C_i$ is an $[n,k_i]_q$-code for $i=1,2$, and $k_1\leq k_2$, then
	$$\dim\hull_{C_2}(C_1) \leq \max\mathrm{wt}((C_1\star C_2)^\perp) - n + k_1.$$
\end{proposition}

\begin{proof}
Let $G_1$ and $G_2$ be generator matrices of $C_1$ and $C_2$, respectively, such that
$$G_1G_2^T=\begin{pmatrix} 0& 0\\ 0&I_\ell\end{pmatrix},$$
	
	\noindent where $\ell$ is defined as $k_1-\dim\hull_{C_2}(C_1)$. Since a basis for $C_1\star C_2$ is given by the set
	$\{{\rm Row}_i(G_1)~\star~{\rm Row}_j(G_2) \ : \ i=1,\ldots,k_1, j=1,\ldots,k_2\}$, then $\lambda=
 \sum_{i \in [n-l]} e_i \in (C_1\star C_2)^\perp$, and  the conclusion follows. 
\end{proof}

The summary of these results is the following theorem.

\begin{theorem}\label{22.11.05}
Let $C_i$ be an $[n,k_i]_q$-code with $q > 2$ for $i=1,2$. For any $\ell$ with $\max\{0,k_1-k_2\}\leq \ell \leq \max \mathrm{wt}((C_1\star C_2)^\perp)-n+k_1$, there exists a code $C_{1,\ell}$ equivalent to $C_1$ such that
 $$\dim\hull_{C_{2}} (C_{1,\ell}) =\ell.$$
In particular, if $\max\mathrm{wt}((C_1\star C_2)^\perp)=\min\{n,2n-k_2-k_1\}$, $\ell$ runs over all the possible values of $\dim\hull_{C_2}(C_1^\prime)$, where $C_1^\prime$ is a code equivalent to $C_1$.
\end{theorem}

\begin{proof}
	The result follows from Proposition \ref{Prop:maxwt}, Theorem \ref{22.10.29}, and Lemma \ref{22.11.03}.
\end{proof}

\begin{remark}\rm
We remark that an algorithm for increasing the relative hull would require finding a codeword in $(C_1\star C_2)^\perp$ of appropriate weight. 
Provided such a word can be found, one can implement an algorithm similar to Algorithm \ref{alg:1}.

\end{remark}

We can find a worse but easier-to-compute lower bound on the maximum rank of the relative hull by using a bound from~\cite{RAVAGNANI20161946} on optimal anticodes. 

\begin{lemma}\cite{RAVAGNANI20161946}
\label{Prop:anticodebound}
If $C \subseteq \mathbb{F}_q^n$ is a linear code, then $\dim_{\mathbb{F}_q}(C) \leq \max\mathrm{wt}(C)$.
\end{lemma}

A code $C \subseteq \mathbb{F}_q^n$ with $\dim_{\mathbb{F}_q}(C) = \max\mathrm{wt}(C)$ is said to be an {\it optimal linear anticode}.

\begin{corollary}\label{Cor:maxdim}
If $C_i$ is an $[n,k_i]_q$-code for $i=1,2$, and $k_1\leq k_2$, then
\[ \max\{\dim \hull_{C_2}(C_1D_{\lambda}) \ |\ {\lambda} \in (\mathbb{F}_q^\ast)^n\} \geq k_1- \dim(C_1\star C_2).\]
\end{corollary}

\begin{proof}
By Lemma~\ref{Prop:anticodebound}, $\dim(C_1\star C_2)^\perp\leq\max\mathrm{wt}(C_1\star C_2)^\perp$.
Thus, Proposition~\ref{Prop:maxwt} gives the conclusion.
\end{proof}

\begin{remark}
Assume that $q\neq 2.$ An optimal anticode of dimension $k$ is permutation equivalent  to $\mathbb{F}_q^k\oplus \{0\}^{n-k}$; see~\cite{RAVAGNANI20161946} for details. Moreover, the dual of an optimal anticode is an optimal anticode. Consequently, the bound in Corollary~\ref{Cor:maxdim} can only be met if $(C_1\star C_2)^\perp$ is an optimal anticode, which implies $C_1\star C_2$ is an optimal anticode. Thus, the minimum rank of $G_1D_{\lambda} G_2^T$ equals the maximum weight of $C_1\star C_2$.
\end{remark}


\section{Applications to quantum codes} \label{quantum_section}
Many quantum code constructions focus on creating codes that do not require entanglement assistance or pairs of maximally entangled quantum states. However, more recently, propagation rules to construct quantum codes have been established ~\cite{Grassl22, Luo22}. Luo, Ezerman, Grassl, and Ling constructed in~\cite{Grassl22} new quantum codes with reduced length by increasing the parameter $c$ and using the Hermitian construction of Theorem~\ref{22.10.01}. Luo, Ezerman, and Ling also gave three new propagation rules related to entanglement using the Hermitian construction in~\cite{Luo22}.  The first rule increases the parameter $c$ while increasing the dimension, the second rule keeps $c$ unchanged while increasing the length, and the third rule decreases $c$ while increasing the length. 

We now state some results that are consequences of the previous sections.
\begin{theorem}
Let $C_i$ be an $[n,k_i]_q$-code for $i=1,2$, with $q > 2$ and $k_1\leq k_2$.
For any integer $c$ with $k_1- \dim(C_1\cap C_2^\perp) \leq c \leq k_1$,
there is an $[[n,\kappa,\delta;c]]_q$ quantum code $Q$ with
    $$\kappa=n-k_1-k_2+c\quad \text{and} \quad
    \delta\geq\min\{d(C_1^\perp),d(C_2^\perp)\}.$$
Moreover, if $\delta=\min\{d(C_1^\perp),d(C_2^\perp)\}$, then $Q$ is pure. 
\end{theorem}
\begin{proof}
We obtain the result using Theorem~\ref{22.10.29} and the CSS construction given in Theorem~\ref{22.10.01}.
\end{proof}

Let $Q$ be the quantum code obtained via the CSS construction using $C_1$ and $C_2$ and
$\delta(Q)=\min\{\wt(C_1^\perp\setminus C_2),\wt(C_2^\perp\setminus C_1)\}$, where we denote $C_1^\perp \setminus (C_2\cap C_1^\perp)$ by $ C_1^\perp \setminus C_2$ for the sake of simplicity. In general, if we take the quantum code $Q'$ constructed via the CSS construction using $C_1$ and $C^\prime_2$, where $C^\prime_2$ is equivalent to $C_2$ and $C_1\cap {C^\prime_2}^\perp=\{0\}$, then $\delta(Q')=\min\{\wt(C_1^\perp\setminus {C'_2}^\perp),d({C'_2}^\perp)\}$. If $Q$ is not pure, it is possible that $\delta(Q)\geq\delta(Q')$ since the equivalence can worsen the minimum distance. Otherwise, we have the following result.
\begin{proposition}
Let $Q$ be the pure quantum code obtained via the CSS construction using $C_1$ and $C_2$. If $Q'$ is a quantum code obtained via the CSS construction using $C_1$ and a monomially equivalent code $C'_2$ to $C_2$, then $\delta(Q')\geq\delta(Q)$.
\end{proposition}
\begin{proof}
As $Q$ is pure, we obtain that $\delta(Q)=\min\{d(C_1^\perp),d(C_2^\perp)\}$. Note that
$\delta(Q')=\min\{\wt(C_1^\perp\setminus C_2'),\wt(C_2'^\perp\setminus C_1)\} \geq
\min\{d(C_1^\perp),d(C_2'^\perp)\} = \delta(Q)$. Thus, the result follows.
\end{proof}

If $d(C_1^\perp)<d(C_2^\perp)$, the equality in the previous corollary depends on how many minimum weight codewords of $C_1^\perp$  are outside $C_2$. If any {code equivalent to} $C_2$ does not contain all minimum weight codewords of $C_1^\perp$, then the purity is preserved. The following corollary provides an instance of such constructions.

\begin{proposition}\label{corollary:pure}
Let $Q$ be the pure quantum code obtained via the CSS construction using $C_1$ and $C_2$.
Assume one of the following conditions holds:
 \begin{enumerate}
 	\item $d(C_1^\perp)<\min\{d(C_2),d(C_2^\perp)\}$.
 	\item $d(C_1^\perp)=d(C_2^\perp)$ and $d(C_1^\perp)<\min\{d(C_1),d(C_2)\}$. 
  \end{enumerate} 
Then, any quantum code $Q'$ constructed via the CSS construction using $C_1$ and an equivalent code $C'_2$ to $C_2$ is pure and $\delta(Q')=\delta(Q)=d(C_1^\perp)$.
\end{proposition}
\begin{proof}
As $Q$ is pure, we obtain that $\delta(Q)=\min\{d(C_1^\perp),d(C_2^\perp)\}$.
Note that
$\delta(Q')=\min\{\wt(C_1^\perp\setminus C_2'),\wt(C_2'^\perp\setminus C_1)\}$ and $d(C_2)=d(C_2')$.

Assume {\it (1)}. As $d(C_1^\perp)<d(C_2)$, all codewords of minimum weight in $C_2'$ are outside of $C_1'$. Thus, $\wt(C_1^\perp\setminus C_2') = d(C_1^\perp)$. As $d(C_1^\perp)< d(C_2^\perp) = d(C_2'^\perp) < d(C_2'^\perp\setminus C_1)$,
we obtain $\delta(Q')=\min\{\wt(C_1^\perp\setminus C_2'),\wt(C_2'^\perp\setminus C_1)\}= d(C_1^\perp)=\delta(Q)$.

Assume {\it (2)}. As $d(C_1^\perp)<d(C_2)$, all codewords of minimum weight in $C_2'$ are outside of $C_1'$. Thus, $\wt(C_1^\perp\setminus C_2') = d(C_1^\perp)$.
As $d(C_2^\perp)<d(C_1)$, then all codewords of minimum weight in $C_1'$ are outside of $C_2'$. Thus, $\wt(C_2^\perp\setminus C_1') = d(C_2^\perp)$. We obtain $\delta(Q')=\min\{\wt(C_1^\perp\setminus C_2'),\wt(C_2'^\perp\setminus C_1)\}=\min\{d(C_1^\perp),d(C_2'^\perp)\}= d(C_1^\perp)=\delta(Q)$.
\end{proof}

\begin{example}
Let $\mathcal{S}=S_1\times S_2\subseteq\mathbb{F}_q^2$ and $g(x,y)=g_1(x)g_2(y)\in\mathbb{F}_q[x,y]$, where $g(s_1,s_2)\neq 0$ for all $(s_1,s_2)\in\mathcal{S}$. Define the tensor product 
$$T(\mathcal{S},g)=\mathrm{RS}(S_1,g_1)\otimes\mathrm{RS}(S_2,g_2),$$
where $\mathrm{RS}(S_i,g_i)=\{(f(s)/g_i(s))_{s\in S_i} \ |\ f\in\mathbb{F}_q[x],\ \deg f < \deg g_i\}$ for $i=1,2$. Note that $\mathrm{RS}(S_i,g_i)$ is a generalized Reed-Solomon code with evaluation points in $S_i$, dimension $\deg(g_i)$, and multipliers $1/g_i(s),$ $s\in S_i$. In~\cite{multigoppa}, the authors used the codes $T(\mathcal{S},g)$ to build entanglement-assisted quantum error-correcting codes with new parameters with respect to the literature. In Table \ref{table}, we build LCD codes exhibiting the same set of parameters. But then, by computing the dual of the square (using ~\cite{magma}), we prove that there is a $\lambda\in(\mathbb{F}_q^{\ast})^n$ such that $C^\perp\subseteq \lambda\star C$ for any of these LCD codes. Thus, Proposition \ref{23.06.13} enables us to increase the hull, and Theorem \ref{22.10.29} allows us to vary the parameter $c$ between $0$ and $n-k$, where $k$ is the dimension of the code. Other works related to tensor products and quantum codes are \cite{10.5555/3179575.3179578, 1523493, 8277961}.

\begin{table}[h!]
\begin{center}
\resizebox{\columnwidth}{!}{
\begin{tabular}{|c|c|c|c|c|c|c|}
\hline
\multirow{2}{*}{Field} &\multirow{2}{*}{$\cal S$}& \multirow{2}{*}{$g(x,y)$} & Puncturing in & \multirow{2}{*}{Parameters} & \multirow{2}{*}{Values for $h$} \\
&& & the following entries && \\
\hline
\hline
$\F_8$ &$\F_8 \times \{a^1,a^2 \}$& $\left(x^2+x+a^5\right)\left(y\right)$ & $\left\{8,\ldots,15\right\}$ & $[[8,2-h,6;6-h]]_{8}$&$0\leq h\leq 2$ \\
\hline
$\F_8$ &$\F_8 \times \{a^1,a^2 \}$& $\left(x^2+x+a^5\right)\left(y\right)$ & $\left\{10,\ldots,16\right\}$ & $[[9,2-h,7;7-h]]_{8}$&$0\leq h\leq 2$ \\
\hline
$\F_8$ &$\F_8 \times \{a^1,a^2 \}$& $\left(x^2+x+a^5\right)\left(y\right)$ & $\left\{11,\ldots,16\right\}$ & $[[10,2-h,8;8-h]]_{8}$&$0\leq h\leq 2$ \\
\hline
$\F_8$ &$\F_8 \times \{a^1,a^2 \}$& $\left(x^2+x+a^5\right)\left(y\right)$ & $\left\{12,\ldots,16\right\}$ & $[[11,2-h,9;9-h]]_{8}$&$0\leq h\leq 2$ \\
\hline
$\F_{16}$ &$\F_{16} \times \{a^1,a^2 \}$& $\left(x^2+x+a^3\right)\left(y\right)$ & $\left\{19,\ldots,32\right\}$ & $[[18,2-h,16;16-h]]_{16}$&$0\leq h\leq 2$ \\
\hline
$\F_{16}$ &$\F_{16} \times \{a^1,a^2 \}$& $\left(x^2+x+a^3\right)\left(y\right)$ & $\left\{21,\ldots,32\right\}$ & $[[20,2-h,18;18-h]]_{16}$&$0\leq h\leq 2$ \\
\hline
$\F_{16}$ &$\F_{16} \times \{a^1,a^2 \}$& $\left(x^2+x+a^3\right)\left(y\right)$ & $\left\{23,\ldots,32\right\}$ & $[[22,2-h,20;20-h]]_{16}$&$0\leq h\leq 2$ \\
\hline
$\F_{16}$ &$\F_{16} \times \{a^1,a^2 \}$& $\left(x^3+a\right)\left(y\right)$ & $\left\{26,\ldots,32\right\}$ & $[[25,3-h,21;20-h]]_{16}$&$0\leq h\leq 3$ \\
\hline
$\F_{16}$ &$\F_{16} \times \{a^1,a^2 \}$& $\left(x^3+a\right)\left(y\right)$ & $\left\{28,\ldots,32\right\}$ & $[[27,3-h,23;24-h]]_{16}$&$0\leq h\leq 3$ \\
\hline
$\F_{16}$ &$\F_{16} \times \{a^1,a^2 \}$& $\left(x^3+a\right)\left(y\right)$ & $\left\{30,\ldots,32\right\}$ & $[[29,3-h,25;26-h]]_{16}$&$0\leq h\leq 3$ \\
\hline
$\F_{16}$ &$\F_{16} \times \{a^1,a^2 \}$& $\left(x^3+a\right)\left(y\right)$ & $\left\{32\right\}$ & $[[31,3-h,27;28-h]]_{16}$&$0\leq h\leq 3$ \\
\hline
$\F_{25}$ &$\F_{25} \times \{a^1,a^2,a^3 \}$& $\left(x^3+a\right)\left(y\right)$ & $\left\{60,\ldots,75\right\}$ & $[[59,3-h,53;56-h]]_{25}$&$0\leq h\leq 3$ \\
\hline
$\F_{49}$ &$\F_{49} \times \{a^1,\ldots,a^4 \}$& $\left(x^3+a\right)\left(y\right)$ & $\left\{168,\ldots,196\right\}$ & $[[167,3-h,159;164-h]]_{49}$&$0\leq h\leq 3$ \\
\hline
$\F_{49}$ &$\F_{49} \times \{a^1,\ldots,a^4 \}$& $\left(x^3+a\right)\left(y\right)$ & $\left\{175,\ldots,196\right\}$ & $[[174,3-h,166;171-h]]_{49}$&$0\leq h\leq 3$ \\
\hline
\end{tabular}}
\end{center}
\caption{New EAQECCs. Here, $\F_{q}^*=\left< a\right>$ for every row; the elements of $\F_{q}$ are ordered $0,a^0,\ldots,a^{q-2}$; the elements of  $\mathcal{S}=\F_{q}\times \{a^1,a^2,\ldots,a^i \}$ are ordered  by $(0,a^1),(a^0,a^1),\ldots,(a^{q-2},a^2),\ldots,(0,a^i),(a^0,a^i),\ldots,(a^{q-2},a^i)$; and generator matrix columns are ordered using the elements in $\mathcal{S}$.}
\label{table}
\end{table}
\end{example}

Table~\ref{table} shows that by puncturing $T(\mathcal{S},g)$, which is the dual of a multivariate Goppa code~\cite{multigoppa}, and using Theorem \ref{22.10.29}, we can fill in some gaps or improve the minimum distance or the dimension of some of the best-known EAQECCs recently published by L. Sok~\cite{LSok}. Other recent related work appears in \cite{fan,sok2}.

We now show the existence of entanglement-assisted quantum MDS codes for $q > 2$ and $1< n \leq q+1$. An $[[n,\kappa,\delta;c]]_q$-quantum code with $\delta-1\leq\frac{n}{2}$ satisfying 
$$2(\delta-1)=n-\kappa+c$$
is called an {\it EAQMDS code}. EAQMDS codes for $\delta>\frac{n}{2}+1$ exist, but since we are considering codes derived from the CSS Construction, we are concerned about codes with the mentioned restriction. For more on the quantum Singleton type bounds and EAQMDS codes, see \cite{entropic}.

Constructions in Theorem \ref{22.10.01} and \ref{22.10.02} give rise to EAQECCs codes if $C_1$ and $C_2$ are MDS codes of the same rate in the CSS construction, or $C$ is a Hermitian MDS code. Many constructions for EAQMDS codes have relied on the CSS or the Hermitian constructions, so there is a vast literature on how to find MDS codes with specific Euclidean, Hermitian, or Galois hull~\cite{mdsgalois, euhermds, mdsgrs, mdshulls, mdsgrs2}. Table~\ref{table2} exhibits some of the EAQMDS codes previously reported, which were based on the possibility of finding a proper isometry of an MDS code to get $\mathrm{rank}(G I_{\lambda^2} G^T)=k-h$, where $G$ is a generator matrix. These results complement those on unassisted ($c=0$) quantum MDS codes~\cite{smallfields, qrm}. As a generalization, we get the following result as a consequence of Theorem \ref{22.11.05}.  

\begin{theorem}\label{theo.mds}
If $q > 2$, $1 < n \leq q+1$, and $1\leq k\leq n/2$, then there is an $$[[n,n-k-h,k+1;k-h]]_q$$ EAQMDS code for any $0\leq h\leq k$.
\end{theorem}

\begin{proof}
	Let $C$ be a (possibly extended or double extended) generalized Reed-Solomon code of dimension $k$. It is known that $C^\perp$ is a generalized Reed-Solomon code of dimension $n-k$. Thus, there is $\lambda\in(\mathbb{F}_q^\ast)^n$ such that $C\subseteq (\lambda\star C)^\perp$, or equivalently, $\dim\hull_{\lambda\ast C}(C)=k$. Applying Theorem \ref{22.10.29} to $C_1=C$ and $C_2=\lambda\star C$, we get the result. 
\end{proof}

\begin{table}[ht]
\begin{center}

\begin{tabular}{|m{.8\textwidth}|c|}
\hline
Conditions & Reference\\\hline
$q>3$, $k\leq m\leq n/2$, and exists a self-orthogonal  $[n,m]$ GRS code.&\cite{mdsgrs}\\\hline

$q>3$, $n<q$, and exists a self-orthogonal $[n+1,k]$ extended GRS code.&\cite{mdsgrs}\\\hline

$q=p^m$, $e\leq m$, $n|(q-1)$ and $k\leq\frac{p^e+n-1}{p^e+1}$ or $n|(p^e-1)$&\cite{mdsgalois}\\\hline

$q=p^m$ odd, $e\leq m-1$, $n\leq p^e$, and $2e|m$. &\cite{mdsgalois}\\\hline

$q=p^m>3$, $p$ odd prime, $n=p^r$, $r|m$, and $2n-k-q-2\geq h\geq 0$. &\cite{mdsgrs2}\\\hline

$q=p^m>3$, $p$ odd prime, $p|n$, $(n-1)|(q-1)$, and $2n-q<k+1$. &\cite{mdsgrs2}\\\hline

$q > 2$ even and $1<n\leq q+1$. &\cite{mdshulls}\\ \hline 

$q > 3$ odd, $n=q+1$, and $k=\frac{q+1}{2}$. &\cite{mdshulls}\\\hline

$q > 3$ odd, $n>2$, $(n-1)|(q-1)$, and $-(n-1)$ is a square in $\mathbb{F}_q$. &\cite{mdshulls}\\\hline

$q > 2$ and $1 < n \leq q+1$. & Theorem \ref{theo.mds}\\\hline

\end{tabular}
\end{center}
\caption{Conditions that guarantee the existence of an $[[n,n-~k-~h,k+~1;k-h]]_q$ EAQMDS code for $k\leq n/2$ and for any $0\leq h\leq k$.}
\label{table2}
\end{table}



\begin{remark}\rm
For $k > n/2$, we have a similar result to Theorem~\ref{theo.mds}. In fact,
if $q > 2$, $1 < n \leq q+1$, and $k > n/2$, then there is an $$[[n,n-k-h,k+1;k-h]]_q$$ EAQECC code for any $0\leq h\leq k$, but this quantum code is not necessarily an EAQMDS code.
\end{remark}

Theorem \ref{theo.mds} can also be extended to other families of QMDS codes ($c=0$) built with the Hermitian construction. Indeed, by reducing the Hermitian hull, the existence of an EAQMDS of length $n\leq q^2+1$ can be derived from the existence of a Hermitian self-orthogonal MDS code (see \cite{Grassl22}). Such MDS codes have been reported in~\cite{smallfields, qrm}. Since QMDS are known to be pure \cite{ketkar}, we can apply the propagation rules in \cite{entropic} to puncture QMDS with no assistance to get EAQMDS codes of shorter lengths.

\section{Final Remarks} \label{conclusion_section}
Given two codes $C_1$ and $C_2$, we studied the relative hull of $C_1$ with respect to $C_2$, which is the intersection $C_1\cap C_2^\perp$. We showed that the $e$-Galois relative hull is a particular case of the Euclidean relative hull. We proved that the dimension of the relative hull can always be repeatedly reduced by one by replacing any of the two codes with a monomially equivalent one. The proof illustrates and explains how to construct such an equivalent code. Similarly, we gave conditions under which the dimension of the relative hull can be increased by one via equivalent codes. We showed some consequences of the relative hull on quantum codes and proved the existence of some quantum MDS codes via the CSS construction.

\section*{Acknowledgements} 

This material is based upon work supported by the National Science Foundation under Grant No. DMS-1929284 while the authors were in residence at the Institute for Computational and Experimental Research in Mathematics in Providence, RI, or participated remotely during the Collaborate@ICERM Quantum Error Correction program. We also thank Rodrigo San-Jos\'e for their comments and suggestions on this article.

\bibliography{bib}{}

\begin{thebibliography}{10}

\bibitem{ashikhmin_knill}
A.~Ashikhmin and E.~Knill.
\newblock Nonbinary quantum stabilizer codes.
\newblock {\em IEEE Transactions on Information Theory}, 47(7):3065--3072,
  2001.

\bibitem{macaulay2}
T.~Ball, E.~Camps, H.~Chimal-Dzul, D.~Jaramillo-Velez, H.~L{\'o}pez,
  N.~Nichols, M.~Perkins, I.~Soprunov, G.~Vera-Mart{\'\i}nez, and G.~Whieldon.
\newblock Coding theory package for {M}acaulay2.
\newblock {\em Journal of Software for Algebra and Geometry}, 11(1):113--122,
  2022.

\bibitem{magma}
W.~Bosma, J.~Cannon, and C.~Playoust.
\newblock The {M}agma {A}lgebra {S}ystem {I}: The user language.
\newblock {\em Journal of Symbolic Computation}, 24(3):235--265, 1997.

\bibitem{brun_science}
T.~Brun, I.~Devetak, and M.-H. Hsieh.
\newblock Correcting quantum errors with entanglement.
\newblock {\em Science}, 314(5798):436--439, 2006.

\bibitem{CRSS}
A.~Calderbank, E.~Rains, P.~Shor, and N.~Sloane.
\newblock Quantum error correction via codes over {GF}(4).
\newblock {\em IEEE Transactions on Information Theory}, 44(4):1369--1387,
  1998.

\bibitem{CS}
A.~R. Calderbank and P.~W. Shor.
\newblock Good quantum error-correcting codes exist.
\newblock {\em Physical Review A}, 54:1098--1105, Aug 1996.

\bibitem{mdsgalois}
M.~Cao.
\newblock {MDS} codes with {G}alois hulls of arbitrary dimensions and the
  related entanglement-assisted quantum error correction.
\newblock {\em IEEE Transactions on Information Theory}, 67(12):7964--7984,
  2021.

\bibitem{lcd}
C.~Carlet, S.~Mesnager, C.~Tang, Y.~Qi, and R.~Pellikaan.
\newblock Linear codes over $\mathbb{F}_q$ are equivalent to {LCD} codes for
  $q>3$.
\newblock {\em IEEE Transactions on Information Theory}, 64(4):3010--3017,
  2018.

\bibitem{10006387}
H.~Chen.
\newblock On the hull-variation problem of equivalent linear codes.
\newblock {\em IEEE Transactions on Information Theory}, 69(5):2911--2922,
  2023.

\bibitem{fan}
J.~Fan, J.~Li, Y.~Zhou, M.-H. Hsieh, and H.~V. Poor.
\newblock Entanglement-assisted concatenated quantum codes.
\newblock {\em Proceedings of the National Academy of Sciences},
  119(24):e2202235119, 2022.

\bibitem{10.5555/3179575.3179578}
J.~Fan, Y.~Li, M.-H. Hsieh, and H.~Chen.
\newblock On quantum tensor product codes.
\newblock {\em Quantum Information \& Computation}, 17(13–14):1105–1122,
  2017.

\bibitem{fan_zhang}
Y.~Fan and L.~Zhang.
\newblock {Galois self-dual constacyclic codes}.
\newblock {\em Designs, Codes and Cryptography}, 84(3):473--492, 2017.

\bibitem{euhermds}
W.~Fang, F.-W. Fu, L.~Li, and S.~Zhu.
\newblock {E}uclidean and {H}ermitian hulls of {MDS} codes and their
  applications to {EAQECC}s.
\newblock {\em IEEE Transactions on Information Theory}, 66(6):3527--3537,
  2019.

\bibitem{mdsgrs}
X.~Fang, M.~Liu, and J.~Luo.
\newblock On {E}uclidean hulls of {MDS} codes.
\newblock {\em Cryptography and Communications}, 13:1--14, 2021.

\bibitem{GHMR19}
C.~Galindo, F.~Hernando, R.~Matsumoto, and D.~Ruano.
\newblock Entanglement-assisted quantum error-correcting codes over arbitrary
  finite fields.
\newblock {\em Quantum Information Processing}, 18(4):116, 2019.

\bibitem{GaHeMaRu}
C.~Galindo, F.~Hernando, R.~Matsumoto, and D.~Ruano.
\newblock Correction to: Entanglement-assisted quantum error-correcting codes
  over arbitrary finite fields.
\newblock {\em Quantum Information Processing}, 20(6):216, 2021.

\bibitem{entropic}
M.~Grassl, F.~Huber, and A.~Winter.
\newblock Entropic proofs of {S}ingleton bounds for quantum error-correcting
  codes.
\newblock {\em IEEE Transactions on Information Theory}, 68(6):3942--3950,
  2022.

\bibitem{1523493}
M.~Grassl and M.~Rotteler.
\newblock Quantum block and convolutional codes from self-orthogonal product
  codes.
\newblock In {\em Proceedings. International Symposium on Information Theory,
  2005. ISIT 2005.}, pages 1018--1022, 2005.

\bibitem{smallfields}
M.~Grassl and M.~R{\"o}tteler.
\newblock Quantum mds codes over small fields.
\newblock In {\em 2015 IEEE International Symposium on Information Theory
  (ISIT)}, pages 1104--1108. IEEE, 2015.

\bibitem{Mac2}
D.~R. Grayson and M.~E. Stillman.
\newblock {M}acaulay2, a software system for research in algebraic geometry.

\bibitem{guenda}
S.~J. K.~Guenda and T.~A. Gulliver.
\newblock Constructions of good entanglement assisted quantum error correcting
  codes.
\newblock {\em Designs, Codes and Cryptography}, 86(1):121--136, 2018.

\bibitem{ketkar}
A.~Ketkar, A.~Klappenecker, S.~Kumar, and P.~Sarvepalli.
\newblock Nonbinary stabilizer codes over finite fields.
\newblock {\em IEEE Transactions on Information Theory}, 52(11):4892--4914,
  2006.

\bibitem{liu_pan}
H.~Liu and X.~Pan.
\newblock Galois hulls of linear codes over finite fields.
\newblock {\em Designs, Codes and Cryptography}, 88(2):241--255, feb 2020.

\bibitem{eaqec_lcd}
X.~Liu, H.~Liu, and L.~Yu.
\newblock New {EAQEC} codes constructed from {G}alois {LCD} codes.
\newblock {\em Quantum Information Processing}, 1, 2020.

\bibitem{multigoppa}
H.~H. L{\'o}pez and G.~L. Matthews.
\newblock Multivariate {G}oppa codes.
\newblock {\em IEEE Transactions on Information Theory}, 69(1):126--137, 2022.

\bibitem{mdshulls}
G.~Luo, X.~Cao, and X.~Chen.
\newblock {MDS} codes with hulls of arbitrary dimensions and their quantum
  error correction.
\newblock {\em IEEE Transactions on Information Theory}, 65(5):2944--2952,
  2018.

\bibitem{Grassl22}
G.~Luo, M.~F. Ezerman, M.~Grassl, and S.~Ling.
\newblock How much entanglement does a quantum code need?, 2022.

\bibitem{Luo22}
G.~Luo, M.~F. Ezerman, and S.~Ling.
\newblock Entanglement-assisted and subsystem quantum codes: New propagation
  rules and constructions, 2022.

\bibitem{macwilliams1960}
J.~MacWilliams.
\newblock Error-correcting codes for multiple-level transmission.
\newblock {\em Bell System Technical Journal}, 40(1):281--308, 1961.

\bibitem{8277961}
P.~J. Nadkarni and S.~S. Garani.
\newblock Entanglement assisted binary quantum tensor product codes.
\newblock In {\em 2017 IEEE Information Theory Workshop (ITW)}, pages 219--223,
  2017.

\bibitem{rains}
E.~Rains.
\newblock Nonbinary quantum codes.
\newblock {\em IEEE Transactions on Information Theory}, 45(6):1827--1832,
  1999.

\bibitem{RAVAGNANI20161946}
A.~Ravagnani.
\newblock Generalized weights: An anticode approach.
\newblock {\em Journal of Pure and Applied Algebra}, 220(5):1946--1962, 2016.

\bibitem{qrm}
P.~K. Sarvepalli and A.~Klappenecker.
\newblock Nonbinary quantum reed-muller codes.
\newblock In {\em Proceedings. International Symposium on Information Theory,
  2005. ISIT 2005.}, pages 1023--1027. IEEE, 2005.

\bibitem{shibatamatsumoto}
M.~Shibata and R.~Matsumoto.
\newblock Advance sharing of quantum shares for quantum secrets.
\newblock {\em arXiv.2302.14448}, 2023.

\bibitem{sok2}
L.~Sok.
\newblock A new construction of linear codes with one-dimensional hull.
\newblock {\em Designs, Codes and Cryptography}, 2022.

\bibitem{LSok}
L.~Sok.
\newblock On linear codes with one-dimensional {E}uclidean hull and their
  applications to {EAQECC}s.
\newblock {\em IEEE Transactions on Information Theory}, 68(7):4329--4343,
  2022.

\bibitem{Steane}
A.~Steane.
\newblock Multiple-particle interference and quantum error correction.
\newblock {\em Proceedings of the Royal Society of London. Series A:
  Mathematical, Physical and Engineering Sciences}, 452(1954):2551--2577, 1996.

\bibitem{mdsgrs2}
G.~Wang and C.~Tang.
\newblock Application of {GRS} codes to some entanglement-assisted quantum
  {MDS} codes.
\newblock {\em Quantum Information Processing}, 21(3):98, 2022.

\bibitem{wilde}
M.~M. Wilde and T.~A. Brun.
\newblock Optimal entanglement formulas for entanglement-assisted quantum
  coding.
\newblock {\em Physical Review A}, 77:064302, Jun 2008.

\end{thebibliography}
\bibliographystyle{abbrv}

\end{document}